\documentclass[journal,twoside,final]{IEEEtran}

\usepackage{comment}
\usepackage{url}

\usepackage{graphicx} 
\usepackage{subfigure}
\usepackage{amsmath}
\usepackage{amssymb}
\usepackage{bm}
\usepackage{color}
\usepackage{url}

\graphicspath{{./fig/}}

\newcommand{\qb}{{\bf b}}

\newcommand{\qd}{{\bf d}}

\newcommand{\qg}{{\bf g}}
\newcommand{\qh}{{\bf h}}

\newcommand{\qn}{{\bf n}}

\newcommand{\qr}{{\bf r}}

\newcommand{\qt}{{\bf t}}

\newcommand{\qw}{{\bf w}}
\newcommand{\qx}{{\bf x}}

\newcommand{\qA}{{\bf A}}

\newcommand{\qC}{{\bf C}}

\newcommand{\qE}{{\bf E}}

\newcommand{\qH}{{\bf H}}
\newcommand{\qI}{{\bf I}}

\newcommand{\qzero}{{\bf 0}}

\newcommand{\be}{\begin{equation}} \newcommand{\ee}{\end{equation}}
\newcommand{\bea}{\begin{eqnarray}} \newcommand{\eea}{\end{eqnarray}}
\newtheorem{theorem}{Theorem}
\newtheorem{lemma}{Lemma}

\newenvironment{itemize*}%
  {\begin{itemize}%
    \setlength{\itemsep}{0pt}%
    \setlength{\parskip}{0pt}}%
  {\end{itemize}}

\usepackage{flushend}
\renewcommand{\baselinestretch}{0.98} 
\begin{document}
\title{Full-Duplex Cooperative Cognitive Radio with Transmit Imperfections}
\author{Gan Zheng, \IEEEmembership{Senior Member, IEEE}, Ioannis Krikidis, \IEEEmembership{Senior Member, IEEE},\\ and Bj$\ddot{\rm o}$rn Ottersten, \IEEEmembership{Fellow, IEEE}
\thanks{%
Manuscript received September 25, 2012; revised December 31, 2012; accepted February 27, 2013. The associate
editor coordinating the review of this paper and approving it for
publication was Dr. W. Jingxian.}
\thanks{%
G. Zheng and B. Ottersten  are with the
Interdisciplinary Centre for Security, Reliability and Trust (SnT),
  University of Luxembourg, 4 rue Alphonse Weicker,  L-2721
Luxembourg (e-mail: \{gan.zheng, bjorn.ottersten\}@uni.lu). B. Ottersten is also with the Signal Processing Laboratory, ACCESS
Linnaeus Center, KTH Royal Institute of Technology, Sweden (e-mail: bjorn.ottersten@ee.kth.se). }
\thanks{%
I. Krikidis is with the Department of Electrical and Computer Engineering, University of Cyprus, Cyprus (e-mail: krikidis@ucy.ac.cy).}
\thanks{%
Digital Object Identifier 10.1109/TWC.2013.121464.}}

\maketitle

\markboth{IEEE Transactions on Wireless Communications, Vol. XX, No.
XX, Month 2013}{Zheng \MakeLowercase{\textit{et al.}}: Full-Duplex Cooperative Cognitive Radio with Transmit Imperfections}

\pubid{1536-1276/12\$31.00~\copyright~2013 IEEE}

\pubidadjcol

\begin{abstract}
 This paper studies the  cooperation between a primary system and a cognitive system in a cellular network where the cognitive base station
 (CBS) relays the primary signal using  amplify-and-forward or decode-and-forward protocols, and in return it can transmit its own cognitive
 signal. While the commonly used half-duplex (HD) assumption may render the cooperation less efficient due to the two orthogonal channel phases
 employed,  we propose that the CBS can work in a full-duplex (FD) mode to improve the system rate region. The problem of interest is to find the
 achievable primary-cognitive rate region by studying the cognitive rate maximization problem.
  For both modes, we explicitly consider the CBS transmit imperfections, which lead to  the residual self-interference associated with the FD
 operation mode.  We propose  closed-form solutions or efficient algorithms to
 solve the problem when the related residual interference power is non-scalable or scalable with the transmit power. Furthermore, we propose a
 simple hybrid scheme to select the HD or FD mode based on zero-forcing criterion, and provide  insights on the impact of system parameters.
 Numerical results illustrate significant performance improvement by using the FD mode and the hybrid scheme.
\end{abstract}

\begin{keywords}
 Cooperative communications, relay channel, cognitive relaying,
 full-duplex, optimization.
\end{keywords}

\section{Introduction}
\PARstart{R}{ecently}  there has been a new paradigm to improve the spectrum efficiency of a cognitive radio network  by introducing active cooperation between
the primary and cognitive systems \cite{Ephremides-07}-\cite{CCRN-Zheng}. As illustrated in Fig. \ref{fig:ccrn},  the cognitive system helps relay
the traffic from the primary base station (PBS), and in return can utilize the primary spectrum for secondary use. This is of particular importance
to the primary system when the primary user (PU)'s   data rate or outage probability requirement cannot be satisfied by itself. Therefore both
systems have strong incentive to cooperate as long as such an opportunity exists. A three-phase cooperation protocol  between primary and cognitive
systems termed ``spectrum leasing''  is proposed to exploit  primary resources in time and  frequency domain \cite{Spectrum-leasing-Simeone},
\cite{Spectrum-leasing-Su}. During Phase I and II, the cognitive base station (CBS) listens and forwards the primary traffic; in the remaining Phase
III, the CBS can transmit its own signal to the cognitive user (CU). {  Note that to avoid additional interference, in Phase II, the PBS remains idle
and only the CBS transmits signals.} The use of multiple-input-multiple-output (MIMO) antennas and beamforming at the CBS provides additional degree
of freedom for primary-cognitive cooperation in the spatial domain \cite{CCRN-Letaief}\cite{CCRN-Song}\cite{CCRN-Zheng}. { In comparison to the
conventional spectrum leasing, MIMO CBS  requires only two phases: Phase I is the same as that in spectrum leasing while in Phase II, the relay can
both forward primary signal and transmit its own signal.}
 \begin{figure}[t] 
  \centering
  \includegraphics[scale=0.5]{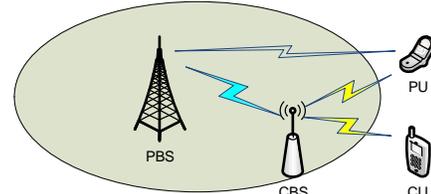}\vspace{-5mm}
  \caption{Cooperation between a primary system  and a cognitive system }\label{fig:ccrn}
  \vspace{-5mm}
\end{figure}
However, most existing works assume half-duplex (HD) mode for the CBS so at least two  orthogonal communication phases are needed,
which brings losses to throughput. As a result, the primary-cognitive cooperation is not always useful, meaning that the achievable primary rate can
be even lower than that of the direct  transmission. To remedy the situation, this paper investigates the potential use of full-duplex (FD) mode for
the CBS, i.e., it simultaneously receives primary message, and  transmits processed primary signal and its own cognitive signal on the same channel.
Since overall only one channel phase is used, the FD mode is an efficient technique to enlarge the achievable rate region. \vspace{-4mm}
\subsection{Related Work}
\pubidadjcol
FD has attracted lots of research interests especially for relay assisted cooperative communication. Traditionally, FD is considered to be infeasible
due to the practical difficulty to recover the desired signal which suffers from the   self-interference from the relay output, which could be as
high as 100 dB \cite{Day-MIMO}. It is shown in \cite{Riihonen-SPAWC-09}   that the FD   relaying in the presence of loop interference is indeed
feasible and can offer higher capacity than the HD mode. Experimental results are reported in  \cite{Duarte-10} that the self-interference can be
sufficiently cancelled to make FD wireless communication feasible in many cases; hardware implementations  in \cite{Sahai} show over 70\% throughput
gains from using the FD over the HD in realistically used cases. Since then, there have been substantial efforts on dealing with self-interference.
Utilizing multiantenna techniques,    \cite{Riihonen-asilomar} proposes to direct the self-interference of a DF relay in the FD mode to the least
harmful spatial dimensions.
 The authors of  \cite{Taneli-Mitigation} analyze a wide range of self-interference mitigation when the relay has multiple antenna, including natural isolation, time-domain cancellation and spatial
domain suppression. The techniques apply to general protocols including amplify-and-forward (AF) and decode-and-forward (DF).
   The transmitter/receiver dynamic-range limitations and channel estimation error at the MIMO DF relay is considered explicitly  in \cite{Day-MIMO-relaying},
and an FD transmission scheme is proposed to maximize a lower bound  of the end-to-end achievable rate by designing transmit covariance matrix.
Considering the tradeoff between residual interference in the FD mode and rate loss in the HD mode, in \cite{Taneli-Hybrid}, a hybrid FD/HD relaying
is proposed together with transmit power adaption to best select the most appropriate mode.  Relay selection is examined in \cite{Krikidis-icc}   in
AF cooperative communication with the FD operation, and shows that the FD relaying results in a zero diversity order despite the relay selection
process.

In the area of cooperative cognitive radio,  there have been very few works on the use of the FD mode. It is worth mentioning that a theoretical
upper-bound for the rate region was found in \cite{Devroye-bound}  \cite{Jafar-bound}\cite{Viswanath-bound}, where the CBS employs dirty paper coding
(DPC) to remove interference from the CU due to the primary signal. However,  DPC requires non-causal information about the primary message at the
CBS, in addition to its implementation complexity; therefore in practice, it is unknown how to achieve this region. FD for CR is first proposed in
\cite{Superposition}  where   the CBS uses AF protocol  and superposition at the CU to improve the rate region. { However, \cite{Superposition}
assumed that at the CBS,  the separation between the transmit and receive antennas is perfect and there is no self-interference, therefore it only
provides a performance upper bound for the FD.}
\subsection{Summary of Contributions}
 The aim of this paper is to study the achievable region using the FD CBS for a cooperative cognitive network taking into account of the
 self-interference.  We assume the primary system is passive, and always tries to operate in its full power. The CBS is  equipped with multiple
 antennas, and is smart enough for forwarding the primary signal, transmitting the cognitive signal and suppressing self-interference. Both AF and DF
 protocols are studied.
 We have made the following
 contributions:
 \begin{itemize}
    \item For CBS operating in the HD mode, we formulate the cognitive  rate maximization problem with constraints on the CBS power and
    the PU rate.  Closed-form solutions are derived.
   \item For CBS operating in the FD mode, we model the self-interference after cancellation due to CBS transmit noise and solve the same problem as the
    HD case for both fixed and scalable transmit noise power. Closed-form solutions are given for the former case and an efficient algorithm is
    developed for the latter by establishing a link between these two
    cases.
    \item We then propose a hybrid HD/FD scheme based on { mode selection and the simplified closed-form zero-forcing (ZF) solutions \cite{kit-zf}
    which nulls out interference between the primary and secondary systems}. Insights are
     given on the impact of system parameters.
    \item Our simulation results demonstrate the enlarged rate region, and substantial performance gain of the proposed FD and hybrid schemes
    compared to the HD mode. It is also verified that the  proposed hybrid scheme performs nearly as well as the best mode selection.
 \end{itemize}
Note that the proposed scheme is not restricted to cellular networks. It can be applied to general cognitive radio scenarios where  secondary
transmitters have multiple antennas with FD capabilities, such as ad hoc cognitive networks  \cite{Chen-06}\cite{Giannakis-11}.
\subsection{Notations}
Throughout this paper, the following notations will be adopted. Vectors and matrices are represented by boldface lowercase and uppercase letters,
respectively.    $\|\cdot\|$ denotes the Frobenius norm. $(\cdot)^\dag$ denotes the Hermitian operation of a vector or matrix.   $\qA\succeq \qzero$
means that $\qA$ is positive semi-definite. $\qI$ denotes an identity matrix of appropriate dimension. Finally, ${\bf x}\sim\mathcal{CN}({\bf
m},{\bf\Theta})$ denotes a vector $\qx$ of complex Gaussian elements with a mean vector of ${\bf m}$ and a covariance matrix of ${\bf\Theta}$.
\section{Baseline HD-CBS System Model and Optimization}
  \begin{figure}[t] 
  \centering
  \includegraphics[width=3in]{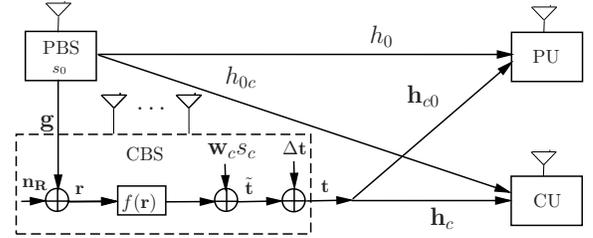}\vspace{-2mm}
  \caption{Cooperating Cognitive System Model in the HD mode}\label{fig:sys}
  \vspace{-5mm}
\end{figure}
\subsection{System Model}
 Consider a cooperative cognitive system shown in Fig. \ref{fig:sys}. The primary system consists of  a single-antenna PBS  and a single-antenna PU.
 The cognitive system includes an $N$-antenna ($N>1$) CBS operating in the HD mode and  a single-antenna CU
 \footnote{
An extension to multiple CUs for the  HD mode can be found in \cite{CCRN-Zheng}.
  Alternatively, interference alignment is a  promising tool in order to control interference in such cooperative cognitive systems.
  }.
It is assumed that cognitive system is time synchronized with the primary network (this assumption holds for all the investigated schemes).
 We assume that the quality of the  primary link is not good
 enough to meet its transmission rate target and the cooperation between the CBS and the PBS becomes necessary \cite{Ephremides-07}.
 To define the system model, we list the following system  parameters:
 \begin{eqnarray}
 h_{0} && \mbox{the scalar channel between the PBS and  the PU;}\nonumber  \\
{h_{0c}} && \mbox{the scalar channel between the PBS and  the CU;}\nonumber\\
 {\qh_{c0}} && \mbox{the $N\times 1$ channel between  the CBS and  the  PU;}\nonumber\\
 {\qh_{c}} && \mbox{the $N\times 1$ channel between  the CBS and   the CU;}\nonumber\\
 \qg && \mbox{the \mbox{$N\times 1$} channel between  the PBS and  the CBS;}\nonumber\\
  n_1 && \mbox{the noise received at PU during Phase I }\nonumber\\
  && \mbox{with $n_1 \sim \mathcal{CN}(0,1)$;}\nonumber\\
   n_2 && \mbox{the noise received at PU during Phase II }\nonumber\\
   && \mbox{with $n_2 \sim \mathcal{CN}(0,1)$;}\nonumber
  \end{eqnarray}
 \begin{eqnarray}
 \qn_R && \mbox{the $N\times 1$ noise vector received at  the  CBS during}\nonumber\\
       && \mbox{Phase I  with $\qn_R\sim \mathcal{CN}(\qzero, \qI)$;}\nonumber\\
  n_c && \mbox{the noise received at the CU during Phase II}\nonumber\\
  && \mbox{with $n_c\sim \mathcal{CN}(0,1)$;}\nonumber\\
  P_0 && \mbox{the available transmit power of  the PBS};\nonumber\\
 P_C && \mbox{the available transmit power  of  the CBS};\nonumber\\
  s_0 && \mbox{the transmit signal for  the PU with $s_0\sim \mathcal{CN}(0,P_0)$;}\nonumber\\
  s_c && \mbox{the transmit signal for the CU with $s_c\sim \mathcal{CN}(0,1)$}.\nonumber
\end{eqnarray}
  All transmit signal, channel and noise elements are assumed to be independent of each other.
{ We assume global perfect channel state information (CSI) is available at the CBS.}
  In the HD mode, the communication takes place in two phases.
 In Phase I, the PBS broadcasts its data $s_0$,  then the received signals at the PU and the CBS are, respectively\footnote{For the sake of presentation, the time slot index is omitted by the
instantaneous expressions of the HD case.},
 \be y_1=   h_0 s_0 + n_{1},~~ \mbox{and}~~\qr=  \qg  s_0 + \qn_R.\ee
 The CBS  processes the received signal and produces $f(\qr)$
 which is defined as

\vspace{-10pt}
{\small
 \be
    f(\qr) = \left\{\begin{array}{ll}
                \qA\qr, & \mbox{for AF where $\qA$ is an $N\times N$ complex}\\
                & \mbox{relay matrix;} \\
                \qw_0 \tilde s_0, & \mbox{for DF where $\qw_0$ is a beamforming vector}\\
                 & \mbox{to forward primary signal and $\tilde s_0 = \frac{s_0}{\sqrt{P_0}}$. }
             \end{array} \right.
 \ee
}%
In Phase II, the CBS superimposes the relaying signal $f(\qr)$ with its own data $s_c$   using the cognitive beamforming vector $\qw_c$, then
transmits to both the PU and the CU. In this phase, the PBS remains idle. The CBS's transmit signal is
 \be
    \tilde \qt = \left\{\begin{array}{ll}
                \qw_c s_c+  \qA\qg   s_0 +  \qA\qn_R, & \mbox{for AF;} \\
                \qw_c s_c + \qw_0 \tilde s_0, & \mbox{for DF. }
             \end{array} \right.
 \ee
 with average power
  \be
    p_R = {\tt E} \|\tilde\qt\|^2= \left\{\begin{array}{ll}
                \|\qw_c\|^2 +  P_0  \| \qA\qg\|^2  +  \|\qA \|^2, & \mbox{for AF;} \\
                \|\qw_c\|^2 + \|\qw_0\|^2, & \mbox{for DF. }
             \end{array} \right.
 \ee
{ To make a fair comparison with the FD mode, we introduce the   transmit noise $\Delta \qt$, which
 combines the effects of phase noise, nonlinear power amplifier, I/Q imbalance,   nonlinear low-noise amplifier and ADC
impairments \cite{Taneli-Mitigation}\cite{fettweis-05}, etc. \footnote{The considered imperfections are general and can also affect all receivers.
Given that the purpose of this work is to study the impact of transmit noise on the FD relaying operation, we assume ideal receivers.} Then the
actually transmitted signal from the CBS is
    \be \qt = \tilde \qt + \Delta
    \qt, ~\Delta\qt \sim \mathcal{CN}(\qzero, P_t\qI), \vspace{-2mm}
 \ee
 where   $P_t$ denotes the transmit noise power and  can either be fixed or scale with $p_R$, depending on how well these impairments are compensated.} It will be seen that in the HD mode, we can use the same approach to solve the problem no matter whether $P_t$ is fixed or not.  While in the
 FD mode using the AF protocol, it makes a difference and we
 will deal with these two cases separately.
 The received signal at the CU is

\vspace{-10pt}
{\small
 \bea
&&y_c =\qh_{c}^\dag\qt + n_c\\
    &&=\left\{\begin{array}{ll}
    \qh_{c}^\dag\qw_c s_c +
     \qh_{c}^\dag\qA\qg s_0 + \qh_{c}^\dag\qA\qn_R + \qh_{c}^\dag\Delta
     \qt+n_{c},&\mbox{for AF;} \\
      \qh_{c}^\dag  \qw_c s_c+  \qh_{c}^\dag \qw_0 \tilde s_0 +\qh_{c}^\dag\Delta \qt  +   n_{c},&\mbox{for DF.}
             \end{array} \right.\notag
\eea
}%
 The received signal-to-interference plus noise ratio (SINR) at CU   is then expressed as
 \be
    \Gamma_c  = \left\{\begin{array}{ll}
    \frac{|\qh_{c}^\dag\qw_c|^2}{    P_0  | \qh_{c}^\dag\qA\qg|^2  +
    \|\qh_{c}^\dag\qA \|^2 + P_t \|\qh_c\|^2 + 1 },& \mbox{for AF;} \\
      \frac{|\qh_{c}^\dag\qw_c|^2}{    |\qh_{c}^\dag\qw_0|^2  +P_t \|\qh_c\|^2+ 1 }, & \mbox{for DF,}
             \end{array} \right.
\ee and the achievable rate is $R_c = \frac{1}{2}\log_2(1+\Gamma_c)$ where the factor $\frac{1}{2}$ arises due to the two orthogonal channel uses.
The received signal at the PU is

{\small
 \bea
&&\hspace{-.65cm} y_2 = \qh_{c0}^\dag\qt+n_2 \\
&&\hspace{-.65cm}=\left\{\begin{array}{ll}
    \qh_{c0}^\dag  \qw_c s_c+ \qh_{c0}^\dag \qA\qg s_0 +  \qh_{c0}^\dag \qA\qn_s +\qh_{c0}^\dag \Delta\qt  +    n_{2},& \mbox{for AF;} \\
     \qh_{c0}^\dag\qw_c s_c  + \qh_{c0}^\dag  \qw_0 s_0+\qh_{c0}^\dag \Delta\qt + n_{2}, & \mbox{for DF.}
             \end{array} \right.\notag
\eea}%
Applying  maximum ratio combining (MRC) to $y_1$ and $y_2$, the received SINR of the PU is the sum  of two channel uses,
 and consequently, the achievable rate is given in (9) at the top of next page.
\begin{figure*}
 {\small
 \be
    R_0 =   \left\{\begin{array}{ll}
    \frac{1}{2}\log_2\left(1+P_0 |{h}_0|^2
    +  \frac{  P_0| \qh_{c0}^\dag\qA\qg|^2}
    { |\qh_{c0}^\dag\qw_c|^2+   \|\qh_{c0}^\dag \qA\|^2 + P_t \|\qh_{c0}\|^2   +
    1}\right),
    \mbox{for AF;}& \\
    \frac{1}{2}\log_2\left(1+ P_0 |h_0|^2 + \min\left(P_0\|\qg\|^2, \frac{|\qh_{c0}^\dag  \qw_0|^2}{|\qh_{c0}^\dag  \qw_c|^2+
      P_t \|\qh_{c0}\|^2+1} \right)\right), 
       \mbox{for DF.}&
             \end{array} \right.
 \ee}
 \hrule
\end{figure*}
\subsection{Achievable Rate Region and Problem Formulation}
 The problem of interest is to find the  rate region   given the primary and cognitive
 power constraints. To achieve this,  we propose to  maximize the CU rate $R_c$ subject to the PU's rate constraint $r_0$ and the CBS's transmit
 power constraint $P_C$, by jointly  optimizing the cognitive beamforming vector $\qw_c$, the relaying processing matrix $\qA$ and the forwarding
 beamforming vector $\qw_0$ for both AF and DF.  Mathematically, the optimization problem can be written as
  \bea
  \max_{}~  R_c ~~~~  \mbox{s.t.} ~~~R_0\ge r_0,~~ p_R\le P_C,
  \eea
  {where the optimization variables are $(\qw_c, \qA)$, $(\qw_c, \qw_0)$ for AF and DF, respectively.}
 Using the monotonicity between the received SINR and the achievable rate, we can derive simplified equivalent problem formulations  for AF and DF.
 More specifically, we have (11) at the top of next page
\begin{figure*}
 {\small
  \bea\label{eqn:prob:rate:max:AF:v0}
    \mbox{P-AF-HD:}~~  \max_{\qw_c,\qA} && \frac{|\qh_{c}^\dag\qw_c|^2}{    P_0  | \qh_{c}^\dag\qA\qg|^2  +
    \|\qh_{c}^\dag\qA \|^2 + P_t \|\qh_c\|^2 + 1 }  \\
    \mbox{s.t.}
   && \frac{ | \qh_{c0}^\dag\qA\qg|^2}
    { |\qh_{c0}^\dag\qw_c|^2+  \|\qh_{c0}^\dag \qA\|^2  + P_t \|\qh_{c0}\|^2   +    1} \ge \gamma_{0_{AF}}^{'},\notag\\
&&   \|\qw_c\|^2 +  P_0  \| \qA\qg\|^2  +  \|\qA \|^2 \le P_C,\notag
 \eea}
 \hrule
\end{figure*}
 and
   \bea\label{eqn:prob:rate:max:DF:v0}
  \mbox{P-DF-HD:} ~~ \max_{\qw_0,\qw_c} && \frac{|\qh_{c}^\dag  \qw_c|^2}{|\qh_{c}^\dag  \qw_0|^2+
 P_t \|\qh_c\|^2+1}   \\
    \mbox{s.t.}
    && P_0\|\qg\|^2 \ge \gamma_{0_{DF}}^{'},\notag\\
   &&   \frac{|\qh_{c0}^\dag  \qw_0|^2}{|\qh_{c0}^\dag  \qw_c|^2+ P_t\|\qh_{c0}\|^2 +1} \ge \gamma_{0_{DF}}^{'},\notag\\
&&\|\qw_c\|^2 +   \| \qw_0\|^2\le P_C, \notag
 \eea
 where $\gamma_{0_{AF}}^{'}\triangleq \frac{2^{2 r_0}-1}{P_0}-|h_0|^2$ and $\gamma_{0_{DF}}^{'}\triangleq  2^{2 r_0}-1 - P_0|{h}_0|^2$.
 Obviously, P-AF-HD appears more complicated than P-DF-HD, so we first focus on P-AF-HD and we later show that actually both problems
 can be solved using the same mechanism.
\vspace{-5mm}
\subsection{The Optimal Structure of $\qA$ in P-AF-HD and Physical Interpretation}
 Problem P-AF-HD involves the optimization of an $N\times N$ matrix $\qA$. In   the following theorem we will characterize the optimal structure of
 $\qA$.
\begin{theorem}\label{eqn:structure:A}
 The optimal AF relay matrix $\qA$ has the structure of
 \be\label{eqn:structure:A}
\qA = \qw_a\qg^\dag,
 \ee
 where   $\bar\qH\triangleq [\qh_{c0} ~\qh_c]$, $\qw_a=\bar\qH\qb$ and $\qb\in \mathbb{C}^{2\times 1}$ are  parameter vectors.
\end{theorem}
The proof is given in Appendix A.  {The structure of $\qA=\qw_a\qg^\dag$  reveals its two functions.
 First it coherently amplifies the received primary signals at different receive antennas using MRC receiver $\qg^\dag$ to produce a noisy version of the primary signal;
 then fowards  the noisy primary signal using  the transmit beamforming vector $\qw_a$.} Theorem 1 not only provides physical interpretation about the above mentioned relay processing but also greatly
 simplifies problem P-AF-HD, which will be seen in the next subsection.
\subsection{Simplified AF Problem}
  It is noted that $\qA$ originally is a general $N\times N$ matrix while (\ref{eqn:structure:A}) indicates that it is actually a rank-1 matrix and
 can be represented by a new vector $\qw_a$.  Employing this structure, the problem P-AF-HD in (\ref{eqn:prob:rate:max:AF:v0}) is simplified to
 \bea\label{eqn:prob:rate:max:AF:v1}
    \min_{\qw_c,\qw_a} && \frac{|\qh_{c}^\dag\qw_c|^2}{ (P_0
\|\qg\|^4 + \|\qg\|^2) |\qh_{c}^\dag\qw_a|^2 +P_t \|\qh_c\|^2+1   }     \\
    \mbox{s.t.}
    &&    \frac{  |\qh_{c0}^\dag\qw_a|^2}    {  |\qh_{c0}^\dag\qw_c|^2 +P_t \|\qh_{c0}\|^2 +1} \ge \frac{\gamma_{0_{AF}}^{'}}{(\|\qg\|^4-\gamma_{0_{AF}}^{'}\|\qg\|^2)},\notag\\
 &&\|\qw_c\|^2 +  (P_0  \|\qg\|^4 +
\|\qg\|^2)\|\qw_a\|^2\le P_C.\notag
 \eea
 For the sake of presentation,  the  closed-form solution to the above problem in a general form is given in Appendix B. From the right hand side of
 the first constraint, it is observed that the problem is feasible or the required PU rate can be satisfied only when
 $\|\qg\|^2>\gamma_{0_{AF}}^{'}$, which complies with the common sense that the end-to-end performance of an AF relay system is upper bounded by the
 channel quality of the PBS-CBS link $\qg$; if this link is too weak, the CBS can not assist the primary transmission.

 For solving problem P-DF-HD in (\ref{eqn:prob:rate:max:DF:v0}),  we can first check whether the first constraint is
 satisfied; if yes, we can remove this constraint and the remaining of the problem has the same structure as (\ref{eqn:prob:rate:max:AF:v1}), then
 its solution is given  in Appendix B; otherwise, the problem is deemed infeasible meaning that the required primary rate cannot be supported even
 with the assistance from the CBS.
\section{FD-CBS System Model and Optimization}
\subsection{System Model}
 In this section, we consider that the  CBS operates in the FD mode, i.e., the CBS can receive  and transmit data  at the same time and frequency. We assume the CBS
has $N_r$ receive antennas/RF chains and  $N_t$ transmit antennas/RF chains. Due to the FD mode, the receive antennas of the CBS will receive a
self-interference from its transmit antennas. The system parameters that are different from the HD based system model are defined as follows:
 \begin{eqnarray}
 {\qh_{c0_{FD}}} && \mbox{the $N_t\times 1$ channel between the CBS and  the PU;}\nonumber\\
 {\qh_{c_{FD}}} && \mbox{the $N_t\times 1$ channel between the CBS and the  CU;}\nonumber\\
 \qg_{FD} && \mbox{\mbox{the $N_r\times 1$} channel between the PBS and the CBS;}\nonumber\\
 \qn_{R_{FD}} && \mbox{the $N_r\times 1$ noise vector received at the CBS}\nonumber\\
   && \mbox{with $\qn_{R_{FD}}\sim \mathcal{CN}(\qzero, \qI)$;}\nonumber\\
  n_0 && \mbox{the noise received at the PU  with $n_0 \sim \mathcal{CN}(0,1)$;}\nonumber\\
 \qH && \mbox{the equivalent $N_r\times N_t$ loop interference  channel }\nonumber\\
 && \mbox{matrix at the CBS;}\nonumber\\
  P_{C_{FD}} && \mbox{the transmit power constraint of the CBS;}\nonumber\\
  P_{0_{FD}} && \mbox{the transmit power  of the PBS}.\nonumber
\end{eqnarray}
We may reuse some variables from the HD mode when no confusion occurs. It is worth noting that $\qH$ is the equivalent loop interference channel
after the self-interference mitigation \cite{Duarte-10}\cite{Sahai}, and its strength depends on the quality of the mitigation process that can be
performed with   techniques such as antenna separation and analog-domain self-interference cancellation, etc.

Although FD and HD are characterized by fundamental operational and complexity differences, the purpose of our study is to provide a fair comparison
between the two relaying modes. This fair comparison is supported by the following assumptions:
 \begin{itemize}
    \item[i)]  We ensure the same  energy consumptions for both modes that is expressed as $P_{C_{FD}}=\frac{P_C}{2}$ and $P_{0_{FD}} = \frac{P_0}{2}$ \cite{Taneli-Hybrid}.
    \item[ii)] We assume the same number of total antennas i.e., an FD system  has $N_t$ transmit and $N_r$ receive antennas and an HD  system with a total of $N = N_r + N_t$ antennas. More specificaly:
    \begin{itemize}
    \item[1)]
         In the HD mode, the CBS has   $N_r$ receive RF chains and $N_t$ transmit RF chains  and  it uses  equal number of antennas:
               $N_r$ receive antennas in Phase I  and $N_t$ transmit antennas in Phase II. In this setting, both the FD and HD systems have the same
               antenna configuration.
    \item[2)]  In addition, we also consider a case that the HD  mode has higher cost:  in the HD mode, the CBS has  $N$ transmit RF chains and $N$ receive RF
    chains, so it can use all $N$  receive antennas in Phase I and all $N$ transmit   antennas in Phase II.
    \end{itemize}
\end{itemize}
     Obviously the performance of 2) is superior to that of 1) but with the cost of extra complexity. We will evaluate the performance of  both cases in
 Section V.

 The details of FD relay processing are illustrated in Fig. \ref{fig:sys:CBS}.
 \begin{figure}[t]
  \centering
  \includegraphics[scale=0.8]{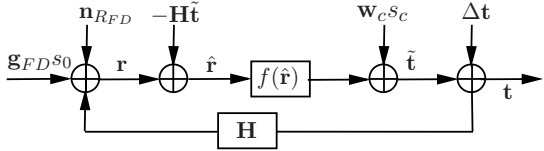}
  \caption{Relay processing at the FD CBS }\label{fig:sys:CBS}
  \vspace{-5mm}
\end{figure}
  We use the index $i$ to denote time instant.  The received signal at the PBS and the CBS are, respectively,
 \bea y[i]&=&   h_0  s_0[i] + \qh_{c0_{FD}}^\dag( \tilde \qt[i] + \Delta\qt[i] )  + n_{0}[i],\\
  \qr[i]&=&  \qg s_0[i] + \qH ( \tilde \qt[i] + \Delta\qt[i])  +
  \qn_{R_{FD}}[i],\eea
where $\tilde \qt[i] + \Delta\qt[i]$ is the actual transmitted signal from the CBS.  The transmit Gaussian noise vector is denoted as $\Delta\qt$ as
the HD mode, $P_t$ denotes the noise power and $\tilde\qt[i]$ is the known transmit signal from the CBS including both relaying primary signal and
cognitive signal:
  \be
      \tilde \qt[i] =  f(\hat \qr(i-D)) +  \qw_c s_c[i],
  \ee
 where $\hat\qr$ is the pre-processed received signal at the CBS and $D>0$ is the processing delay.
 This delay $D>0$ is a general
assumption and refers to the required processing time in order to implement the FD operation \cite{Taneli-Hybrid}. In practical systems the
processing delay for the AF scheme is much smaller than the one in the DF case; however our analysis is general and holds for any $D>0$.

 The relay processing at the CBS is defined as:

\vspace{-10pt}
{\small
  \be
    f(\hat\qr) = \left\{\begin{array}{cl}
                \qA\hat\qr, & \mbox{for AF where   $\qA$ is an $N_t\times N_r$ relay matrix;} \\
                \qw_0 \tilde s_0, & \mbox{for DF where   $\qw_0$ is a beamforming vector,}\\
                 & \mbox{$\tilde s_0=\frac{s_0}{\sqrt{P_{0_FD}}}$. }
             \end{array} \right.
 \ee
 }%
  Although the channel $\qH$ is perfectly estimated at the CBS and the CBS can perfectly remove the noise component $\qH \tilde \qt[i]$,
  but   the term $\qH \tilde \Delta\qt[i]$ remains and forms the residual interference that affects the CBS's input. As a result the CBS gets
 pre-processed signal
 \be
    \hat\qr[i] = \qr[i] - \qH\tilde\qt[i]= \qg s_0[i] + \qH  \Delta\qt[i]   +  \qn_{R_{FD}}[i].
 \ee
  Suppose the CBS further processes $\hat\qr[i]$, then
 \be
    \tilde \qt[i] = \left\{\begin{array}{ll} \qA\hat\qr[i-D] + \qw_c s_c[i]= \qA \qg s_0[i-D] +   \\
    \qA\qH   \Delta\qt[i-D] +
    \qA\qn_r[i-D] +  \qw_c s_c[i], & \\
     \mbox{for AF where $\qA$ is an $N\times N_r$  matrix}; & \\
    \qw_0 \tilde s_0[i-D]+ \qw_c s_c[i], &\\
     \mbox{for DF where~} \tilde s_0[i-D]=\frac{s_0[i-D]}{|s_0[i-D]|}& \end{array} \right.
 \ee
 with average power
 \be p_{R_{FD}} ={\tt E}(\|\tilde \qt\|^2)=\left\{\begin{array}{ll}   P_{0_{FD}}\|\qA \qg_{FD}\|^2   +P_t\|\qA\qH\|^2   +
   \|\qA\|^2  \\ + \|\qw_c\|^2, \mbox{for AF;}& \\
    \|\qw_0\|^2 + \|\qw_c\|^2, \mbox{for DF.}&  \end{array} \right.
 \ee
 The received signal at the   CU is
 {\small
   \bea
      &&y_c[i]  = \qh_{c_{FD}}^\dag (\tilde \qt[i]  + \Delta \qt[i] ) + {h}_{0c} s_0[i] +
    n_c[i] \notag\\
    &&= \left\{\begin{array}{ll}  \qh_{c_{FD}}^\dag \qw_c s_c[i]+ \qh_{c_{FD}}^\dag  \qA \qg s_0[i-D] + \qh_{c_{FD}}^\dag \qA\qH \Delta \qt \notag\\
     +
    \qh_{c_{FD}}^\dag \qA\qn_{R_{FD}}[i-D]  + \qh_{c_{FD}}^\dag\Delta \qt[i]  + h_{0c} s_0[i] +    n_c[i],&\\
     \mbox{for AF;}& \\
    \qh_{c_{FD}}^\dag \qw_c s_c[i]+  \qh_{c_{FD}}^\dag  \qw_0 s_0[i-D] +    \qh_{c_{FD}}^\dag\Delta \qt[i] \notag\\ + h_{0c} s_0[i] +    n_c[i] ,
    &\\
     \mbox{for DF.}& \end{array} \right.
   \eea}
 The received SINR $\Gamma_{c_{FD}}$ at CU is expressed in (\ref{eqn:gam_c_fd}).
 \begin{figure*}
  \be\label{eqn:gam_c_fd}
    \Gamma_{c_{FD}} = \left\{\begin{array}{ll}   \frac{|\qh_{c_{FD}}^\dag \qw_c|^2}{ P_{0_{FD}}|\qh_{c_{FD}}^\dag  \qA \qg|^2 +
     P_t\|\qh_{c_{FD}}^\dag \qA\qH\|^2+ P_t\|\qh_{c_{FD}}\|^2
     +   \|\qh_{c_{FD}}^\dag \qA\|^2 + P_{0_{FD}} |{h}_{0c}|^2 + 1},& \mbox{for AF;} \\
    \frac{|\qh_{c_{FD}}^\dag \qw_c|^2}{  |\qh_{c_{FD}}^\dag  \qw_0|^2 +  P_t\|\qh_{c_{FD}}\|^2       +  P_{0_{FD}} |h_{0c}|^2+ 1} ,
    & \mbox{for DF.} \end{array} \right.
  \ee
  \end{figure*}
  Then the achievable CU rate is $R_{c_{FD}} = \log_2(1+ \Gamma_{c_{FD}})$. {  It is worth nothing that due to the FD, the CU rate does not suffer from the loss
 of a prelog factor $\frac{1}{2}$ observed in the HD mode.} However, due to the residual self-interference, for AF, there is an additional interference term
 $P_t\|\qh_{c_{FD}}^\dag \qA\qH\|^2$ compared to the HD mode.

 The received signal at the PU is shown in (\ref{fig:yi_fd}).
  \begin{figure*}
 \bea\label{fig:yi_fd}
   y[i] &=&  h_0 s_0[i] + \qh_{c0_{FD}}^\dag( \tilde \qt[i]  \Delta \qt[i] )  +
   n_{0}[i]\notag\\
   &=&\left\{\begin{array}{ll}   h_0 s_0[i] + \qh_{c0_{FD}}^\dag  \qA    \qg s_0[i-D] + \qh_{c0_{FD}}^\dag  \qA\qH \Delta \qt[i]\\
    ~~~~~~~+\qh_{c0_{FD}}^\dag \qA \qn_r[i-D] + \qh_{c0_{FD}}^\dag \qw_c s_c[i]    + \qh_{c0_{FD}}^\dag\Delta \qt+
    n_{0}[i],& \mbox{for AF;} \\
     h_0 s_0[i] + \qh_{c0_{FD}}^\dag  \qw_0 s_0[i-D] +\qh_{c0_{FD}}^\dag \qw_c s_c[i]  
      + \qh_{c0_{FD}}^\dag\Delta \qt[i]+   n_{0}[i], & \mbox{for DF.}
      \end{array} \right.
   \eea
   \end{figure*}
  Since the direct link channel $h_0$ is weak, the PU simply treats
  $h_0 s_0[i]$ as noise and decodes $s_0[i-D]$, therefore, the received SINR at the PU is given in (\ref{eqn:gam_0_fd})
   \begin{figure*}
  \be\label{eqn:gam_0_fd}
    \Gamma_{0_{FD}} = \left\{\begin{array}{ll}
    \frac{P_{0_{FD}} |\qh_{c0_{FD}}^\dag  \qA    \qg|^2}{   P_{0_{FD}} |h_0|^2  + P_t \|\qh_{c0_{FD}}^\dag  \qA \qH  \|^2  + P_t\|\qh_{c0_{FD}}\|^2
     +  \| \qh_{c0_{FD}}^\dag  \qA \|^2
      +  |\qh_{c0_{FD}}^\dag \qw_c|^2 + 1},& \mbox{for AF;} \\
    \frac{P_{0_{FD}} |\qh_{c0_{FD}}^\dag \qw_0|^2}{  P_{0_{FD}} |{h}_{0}|^2 + P_t \|\qh_{c0_{FD}}\|^2      + |\qh_{c0_{FD}}^\dag \qw_c|^2 + 1 } ,
     & \mbox{for DF.} \end{array} \right.
  \ee
  \end{figure*}
   The achievable rate for the PU is provide in (\ref{eqn:R_FD}).
    \begin{figure*}
 \be\label{eqn:R_FD}
    R_{0_{FD}}  =   \left\{\begin{array}{ll}
    \log_2\left(1+\frac{P_{0_{FD}} |\qh_{c0_{FD}}^\dag  \qA    \qg|^2}{   P_{0_{FD}} |h_0|^2  + P_t \|\qh_{c0_{FD}}^\dag  \qA \qH  \|^2  +
    P_t\|\qh_{c0_{FD}}\|^2   +  \| \qh_{c0_{FD}}^\dag  \qA \|^2
        +  |\qh_{c0_{FD}}^\dag \qw_c|^2 + 1}\right),& \mbox{for AF;} \\
     \log_2\left(1 + \min\left(P_{0_{FD}}\|\qg_{FD}\|^2, \frac{P_{0_{FD}} |\qh_{c0_{FD}}^\dag \qw_0|^2}{  P_{0_{FD}} |{h}_{0}|^2 + P_t \|\qh_{c0_{FD}}\|^2
          + |\qh_{c0_{FD}}^\dag \qw_c|^2 + 1 }\right)\right), & \mbox{for DF.}
             \end{array} \right.
 \ee
  \end{figure*}
 \subsection{Problem Formulation}
 Similar to the HD case, we will study the achievable rate region by solving the CU SINR maximization problems subject to the PU rate constraint
 $r_0$ and the CBS power constraint $P_{C_{FD}}$. The problems for AF and DF are formulated as (\ref{eqn:prob:SINR:FD:AF:v0})
  \begin{figure*}
 {\small
  \bea\label{eqn:prob:SINR:FD:AF:v0}
    \mbox{P-FD-AF:}&&\notag\\
    ~\max_{\qA,\qw_c} &&   \frac{|\qh_{c_{FD}}^\dag \qw_c|^2}{ P_{0_{FD}}|\qh_{c_{FD}}^\dag  \qA \qg_{FD}|^2 +
     P_t\|\qh_{c_{FD}}^\dag \qA\qH\|^2+ P_t\|\qh_{c_{FD}}\|^2
     +   \|\qh_{c_{FD}}^\dag \qA\|^2 + P_{0_{FD}} |{h}_{0c}|^2 + 1}\\
    \mbox{s.t.} && \frac{P_{0_{FD}} |\qh_{c0_{FD}}^\dag  \qA    \qg|^2}{   P_{0_{FD}} |h_0|^2  + P_t \|\qh_{c0_{FD}}^\dag  \qA \qH  \|^2
     + P_t\|\qh_{c0_{FD}}\|^2   +  \| \qh_{c0_{FD}}^\dag  \qA \|^2
       +  |\qh_{c0_{FD}}^\dag \qw_c|^2 + 1}\notag\\&&\ge \gamma_{0_{FD}} \triangleq 2^{\gamma_0}-1,\notag\\
    && p_{R_{FD}}=P_{0_{FD}}\|\qA \qg_{FD}\|^2   +    P_t\|\qA\qH\|^2   +
   \|\qA\|^2   + \|\qw_c\|^2  \le P_{C_{FD}},\notag
  \eea}
  \end{figure*}
 and (27) on next page.
\begin{figure*}
{\small
    \bea\label{eqn:prob:SINR:FD:DF:v0}
     \mbox{P-FD-DF:}&&\notag\\~\max_{\qw_0,\qw_c} && \frac{|\qh_{c_{FD}}^\dag \qw_c|^2}{  |\qh_{c_{FD}}^\dag  \qw_0|^2 +    P_t \|\qh_{c_{FD}}\|^2
        +  P_{0_{FD}} |{h}_{0c}|^2   +
    1}\\
    \mbox{s.t.} &&
      \frac{P_{0_{FD}}\|\qg_{FD}\|^2}{ P_t   \|\qH\|^2+1}   \ge  \gamma_{0_{FD}} \notag\\
  &&  \frac{|\qh_{c0_{FD}}^\dag \qw_0|^2}{  |\qh_{c0_{FD}}^\dag \qw_c|^2   + P_t   \|\qh_{c0_{FD}}\|^2      +
  P_{0_{FD}} |{h}_{0}|^2  + 1 } \ge \gamma_{0_{FD}} \notag\\
    && p_{R_{FD}}=\|\qw_0\|^2   + \|\qw_c\|^2\le  P_{C_{FD}}.\notag
  \eea}
\end{figure*}
 Comparing to P-HD-AF, it can be checked that the optimal relay processing matrix $\qA$ in P-FD-AF possesses the same structure $\qA =
 \qw_a\qg_{FD}^\dag$ as HD in Theorem 1 and as a result, the problem P-FD-AF is reformulated  as (\ref{eqn:prob:SINR:FD:AF:v1}).
  \begin{figure*}
   \bea\label{eqn:prob:SINR:FD:AF:v1}
    \max_{\qw_a,\qw_c} && \frac{|\qh_{c_{FD}}^\dag \qw_c|^2}{  (1+P_{0_{FD}} |h_{0c}|^2+P_t \|\qh_{c_{FD}}\|^2 ) + (P_{0_{FD}}\|\qg_{FD}\|^4 +   P_t\|\qg_{FD}^\dag\qH\|^2 + \|\qg_{FD}\|^2 )
    |\qh_{c_{FD}}^\dag\qw_a\|^2   }\\
    \mbox{s.t.} && \frac{    |\qh_{c0_{FD}}^\dag  \qw_a|^2}{ (1+P_t \|\qh_{c0_{FD}}\|^2 +  P_{0_{FD}} |h_0|^2 )        +  |\qh_{c0_{FD}}^\dag \qw_c|^2    }\ge
     \frac{\gamma_{0_{FD}} }{(P_{0_{FD}}\|\qg_{FD}\|^4  -\gamma_{0_{FD}} (\|\qg_{FD}\|^2+P_t\|\qg_{FD}^\dag\qH\|^2) )}\notag\\
    && (P_{0_{FD}}\|\qg_{FD}\|^4 + \|\qg_{FD}\|^2   +   P_t\|\qg_{FD}^\dag\qH\|^2 ) \|\qw_a\|^2   + \|\qw_c\|^2  \le
    P_{C_{FD}}.\notag
  \eea
  \end{figure*}
\subsection{Fixed Transmit Noise $P_t$}
 In our previous discussions, we have assumed that $P_t$ is fixed which corresponds to an efficient interference cancellation process. With this
 assumption, (\ref{eqn:prob:SINR:FD:AF:v1}) has the same structure as (\ref{eqn:prob:rate:max:AF:v1}) for AF, and (\ref{eqn:prob:SINR:FD:DF:v0}) and
 (\ref{eqn:prob:rate:max:DF:v0}) share the same structure for DF,
 therefore all solutions can be found using the approach presented in Appendix B.

\subsection{Scalable Transmit Noise $P_t$}
Although  fixing the transmit noise power $P_t$  simplifies the problem   and the solution,   in practice, it is more feasible to assume that $P_t$
scales with the CBS transmit power, i.e., $P_t=\epsilon^2 p_R$ and $P_t=\epsilon^2 p_{R_{FD}}$ for the HD and FD modes, respectively, where
$\epsilon^2$   is a scaling factor and  denotes the percentage of the transmit noise power to the total CBS transmit power. It  depends on the
hardware impairments and can be assumed that it is small for efficient implementations. In this case, the CBS may not use its full power since more
power brings more noise to the receivers in both modes and more self-interference in the FD mode.

Notice that in the HD mode, we have not discussed this issue for the problem formulations   P-DF-HD in (\ref{eqn:prob:rate:max:DF:v0}) and P-AF-HD
in (\ref{eqn:prob:rate:max:AF:v1}), respectively, and the reason is as follows. When $P_t$ scales with $p_R$, the objective functions in both
 (\ref{eqn:prob:rate:max:DF:v0}) and (\ref{eqn:prob:rate:max:AF:v1}) are  non-decreasing functions of $p_R$, which means the relay should always use
the maximum transmit power $p_R = P_C$. We illustrate this by taking the problem P-AF-HD for example. With $P_t=\epsilon^2 p_R$, it becomes

\vspace{-10pt}
{\small
  \bea\label{eqn:prob:rate:max:AF:v:scale}
  \max_{\qw_c,\qA,p_R} && \frac{|\qh_{c}^\dag\qw_c|^2}{    P_0  | \qh_{c}^\dag\qA\qg|^2  +
    \|\qh_{c}^\dag\qA \|^2 + \epsilon^2 p_{R} \|\qh_c\|^2 + 1 }  \\
    \mbox{s.t.}
   && \frac{ | \qh_{c0}^\dag\qA\qg|^2}
    { |\qh_{c0}^\dag\qw_c|^2+  \|\qh_{c0}^\dag \qA\|^2  + \epsilon^2 p_{R} \|\qh_{c0}\|^2   +    1} \ge \gamma_{0_{AF}}^{'},\notag\\
&&   p_R=\|\qw_c\|^2 +  P_0  \| \qA\qg\|^2  +  \|\qA \|^2 \le
P_C.\notag
 \eea
}%
 {Suppose its optimal solution is $(\qw_c^*, \qA^*, p_R^*)$ and the corresponding optimal objective value is $\gamma_c^*$.  We assume the CBS does not
 use maximum transmit power, i.e., $p_R^* = \frac{1}{\alpha} P_C, \alpha>1$.  Then we construct another solution $(\sqrt{\alpha}\qw_c^*,
 \sqrt{\alpha}\qA^*, \alpha p_R^*)$, which satisfies both constraints and  gives higher objective value  $\frac{|\qh_{c}^\dag\qw_c^*|^2}{    P_0  |
 \qh_{c}^\dag\qA^*\qg|^2  + \|\qh_{c}^\dag\qA^* \|^2 + \epsilon^2 p_{R}^* \|\qh_c\|^2 + \frac{1}{\alpha}}>\gamma_c^*$. This contradicts the fact that
 $(\qw_c^*, \qA^*, p_R^*)$ is the optimal solution,  therefore it must hold that $p_R=P_C$.}

 For the FD mode,  we next show that the scalable noise does not affect the approach  to solve  the problem  P-FD-DF in
(\ref{eqn:prob:SINR:FD:DF:v0}). With substitution $P_t=\epsilon^2
p_{R_{FD}}$, it is easy to see that the first constraint is
equivalent to \be
   p_{R_{FD}}\le  \frac{\frac{P_{0_{FD}}\|\qg_{FD}\|^2}{\gamma_{0_{FD}}} -
   1}{\epsilon^2\|\qH\|^2}.
\ee
 Then P-FD-DF  becomes (31).
 \begin{figure*}
    {\small \bea\label{eqn:prob:SINR:FD:DF:v1}
  &&\max_{\qw_0,\qw_c, p_{R_{FD}}} ~ \frac{|\qh_{c_{FD}}^\dag \qw_c|^2}{
|\qh_{c_{FD}}^\dag  \qw_0|^2 +    \epsilon^2 p_{R_{FD}} \|\qh_{c_{FD}}\|^2 +  P_{0_{FD}} |{h}_{0c}|^2   +
    1}\notag\\
  && \mbox{s.t.}~ \frac{|\qh_{c0_{FD}}^\dag \qw_0|^2}{  |\qh_{c0_{FD}}^\dag \qw_c|^2   + \epsilon^2
p_{R_{FD}}   \|\qh_{c0_{FD}}\|^2      +  P_{0_{FD}} |{h}_{0}|^2  + 1 } \ge \gamma_{0_{FD}}, \notag\\
    && p_{R_{FD}}=\|\qw_0\|^2   + \|\qw_c\|^2\le  \bar P_{C_{FD}}\notag\\
    && \triangleq  \min \left(P_{C_{FD}},
    \max\left(0,\frac{\frac{P_{0_{FD}}\|\qg_{FD}\|^2}{\gamma_{0_{FD}}} -
   1}{\epsilon^2\|\qH\|^2}\right)\right).
  \eea}
\end{figure*}
 Problem (\ref{eqn:prob:SINR:FD:DF:v1}) has the similar structure as (\ref{eqn:prob:rate:max:AF:v:scale}) and at the optimum, it must hold that
 $p_{R_{FD}}=\bar P_{C_{FD}}$. To summarize, the scalable transmit noise power does not affect the mechanism  to solve the problems for the HD mode
 and the FD mode with   DF relaying protocol.

However, the above remark may not be true for the FD mode when AF protocol is used since the CBS  amplifies the received noise and more transmit
power results in more self-interference. We will study this problem in the remaining of this section.

With substitution  $P_t=\epsilon^2 p_{R_{FD}}$,  problem (\ref{eqn:prob:SINR:FD:AF:v1}) is updated to (\ref{eqn:prob:SINR:FD:AF:v2}).
 \begin{figure*}
{\small   \bea\label{eqn:prob:SINR:FD:AF:v2}
  \mathbb{P}1: &&\max_{\qw_a,\qw_c,  p_{R_{FD}}}  \frac{|\qh_{c_{FD}}^\dag \qw_c|^2}{  (1+P_{0_{FD}} |h_{0c}|^2+\epsilon^2 p_R
  \|\qh_{c_{FD}}\|^2 ) + (P_{0_{FD}}\|\qg_{FD}\|^4 +    \epsilon^2  p_{R_{FD}}\|\qg_{FD}^\dag\qH\|^2 + \|\qg_{FD}\|^2 )
    |\qh_{c_{FD}}^\dag\qw_a\|^2   }\notag\\
\mbox{s.t.} &&\frac{    |\qh_{c0_{FD}}^\dag  \qw_a|^2}{ (1+ \epsilon^2 p_{R_{FD}} \|\qh_{c0_{FD}}\|^2 +  P_{0_{FD}} |h_0|^2 ) + |\qh_{c0_{FD}}^\dag
\qw_c|^2 }\notag \ge \frac{\gamma_{0_{FD}}}{(P_{0_{FD}}\|\qg_{FD}\|^4  -\gamma_{0_{FD}}
    (\|\qg_{FD}\|^2+\epsilon^2  p_{R_{FD}}\|\qg_{FD}^\dag\qH\|^2) )},\notag\\
    &&p_{R_{FD}} = (P_{0_{FD}}\|\qg_{FD}\|^4 + \|\qg_{FD}\|^2   +    \epsilon^2  p_{R_{FD}}\|\qg_{FD}^\dag\qH\|^2 ) \|\qw_a\|^2   + \|\qw_c\|^2  \le
    P_{C_{FD}}.
  \eea}
  \end{figure*}
   Problem $\mathbb{P}1$ is quite complicated since the CBS power constraint is not always active and it involves the product  of two quadratic
  terms. We denote its objective value as a function of available CBS power $P_{C_{FD}}$, i.e., $\Psi(P_{C_{FD}})$. To solve it, we first
 focus on the following problem $\mathbb{P}2$ whose objective is $\Phi(P)$, a function  of a  parameter
  $P$ in (\ref{eqn:prob:SINR:FD:AF:3}).
 \begin{figure*}
  {\small
   \bea\label{eqn:prob:SINR:FD:AF:3}
  \mathbb{P}2:&& \Phi(P)=\max_{\qw_a,\qw_c}\frac{|\qh_{c_{FD}}^\dag \qw_c|^2}{  (1+P_{0_{FD}} |h_{0c}|^2+\epsilon^2P \|\qh_{c_{FD}}\|^2 ) + (P_{0_{FD}}\|\qg_{FD}\|^4 +
   \epsilon^2 P\|\qg_{FD}^\dag\qH\|^2 + \|\qg_{FD}\|^2 )
    |\qh_{c_{FD}}^\dag\qw_a\|^2   }\notag\\
    \mbox{s.t.}&&  \frac{    |\qh_{c0_{FD}}^\dag  \qw_a|^2}{ (1+ \epsilon^2 P\|\qh_{c0_{FD}}\|^2 +  P_{0_{FD}} |h_0|^2 )        +  |\qh_{c0_{FD}}^\dag \qw_c|^2    }\ge
     \frac{\gamma_{0_{FD}}}{(P_{0_{FD}}\|\qg_{FD}\|^4  -\gamma_{0_{FD}}(\|\qg_{FD}\|^2+\epsilon^2P\|\qg_{FD}^\dag\qH\|^2) )},\notag\\
    &&  (P_{0_{FD}}\|\qg_{FD}\|^4 + \|\qg_{FD}\|^2   +    \epsilon^2 P\|\qg_{FD}^\dag\qH\|^2 ) \|\qw_a\|^2   + \|\qw_c\|^2  \le   P.
  \eea}
  \end{figure*}
   The following theorem characterizes  the relation between $\mathbb{P}1$ and
   $\mathbb{P}2$.
   \begin{theorem}
    Assuming that $\mathbb{P}1$ is feasible, it can be solved by considering $\mathbb{P}2$,
    i.e.,
    $\Psi(P_{C_{FD}}) = \max_{0\le P\le P_{C_{FD}}}\Phi(P)$.
    \end{theorem}
  \begin{proof}
  The proof is based on the following two observations.
  \begin{itemize}
  \item [i)] Suppose the  optimal solution to $\mathbb{P}1$ is given by $(\qw_0^*,\qw_c^*,p_R^*)$ and it is easy to see that $\Psi(P_{C_{FD}}) =
   \Phi(p_R^*)\le \max_{0\le P\le P_{C_{FD}}}\Phi(P)$.
   \item [ii)] {  On the other hand, given an input power $P\le P_{C_{FD}}$, we can solve $\mathbb{P}2$ to obtain $(\qw_0^p, \qw_c^p)$ and suppose its objective
   is $\Phi(P)$. Following the same argument for (\ref{eqn:prob:rate:max:AF:v:scale}), we can see that with the optimal solution, the last constraint
   of $\mathbb{P}2$ must be satisfied  with equality. Therefore $(\qw_0^p, \qw_c^p, P)$  is also  a feasible solution to $\mathbb{P}1$, and this
   implies that $\Phi(P)\le \Psi(P_{C_{FD}}), \forall~ 0\le P\le P_C$.}
 \end{itemize}
   Combining the above two facts, we   conclude that $\Psi(P_{C_{FD}})$ equals
   the maximum of $\Phi(P), 0\le P\le P_{C_{FD}}$.
  \end{proof}
 Theorem 2 indicates that  in order
 to solve the difficult problem $\mathbb{P}1$, it suffices to solve $\mathbb{P}2$ by 1-D search of
 $P$.

\subsection{Implementation Issues}

The implementation of the proposed scheme requires that the CBS can track the CBS-CU, CBS-PU, PBS-CU channels as well as the self-interference channel. The estimation of these parameters can be obtained by using appropriate pilot signals that periodically are sent by the terminals.  More specifically,
\begin{itemize}
\item
The (residual) self-interference channel can be estimated based on a pilot sequence that is sent from the CBS in periodical time instances. In
\cite{FD-WiFi}, the authors implement a pilot-based self-interference estimation mechanism for an FD scheme that incorporates analogue and digital
self-interference mitigation.
\item The estimation of the
CBS-CU and PBS-CU channels in a cognitive radio scenario has been proposed in \cite[Sec. V. D]{Viswanath-bound}. Based on that work, the PBS-CU
channel is firstly estimated at the CU by overhearing the primary radio's pilot signal; then is fed back to the CBS by using the CBS-CU link or a
dedicated out-of-band channel. It is worth noting that this operation requires a synchronization of the CU to the primary radio's pilot signal. The
channel CBS-CU can be estimated at the CU by using the cognitive radio's pilot signal and then is fed back to the CU.
\item In addition to the cognitive implementation in \cite{Viswanath-bound}, the proposed scheme requires also the CBS-PU channel; this information can be
obtained by introducing a periodical pilot signal at the PU (for the purposes of the cognitive cooperation) or by employing blind channel estimation
techniques \cite{heath-blind} at the CBS during the PU transmission.
\end{itemize}

It is worth noting that imperfections on the channel estimation result in performance degradation for the proposed scheme. Since the main objective
of this paper is to introduce a new FD-based cooperative scheme in a cognitive radio context,   we assume perfect channel knowledge as in \cite{Viswanath-bound}. Our
work provides useful performance bounds and serves as a  guideline for practical implementations with realistic channel estimation.

\section{Hybrid HD/FD Mode Selection For the CBS} 
 Although the CBS in the FD mode can improve the rate, it introduces an extra self-interference from the relay's output to the relay's input; on the other hand HD is
 not affected by self-interference due to the orthogonal transmission, but it reduces spectral efficiency. Therefore, no mode is always better than
 the other one and a hybrid solution that switches between the two operation modes can provide extra performance gain. To achieve this, one can
 simply solve each problem for the HD and FD modes, and then choose the better one. However, the closed-form solution given in Appendix B is very complex
 and does not give insights on which mode is preferred under different conditions. In this section, we will develop a simple suboptimal solution
 based on the ZF criterion, which will be used for mode selection. Towards this, we
 first state the following lemma:
 \begin{lemma}
    For both HD and FD modes, the  DF relaying protocol achieves
    higher CU rate than using AF protocol when the direct link $h_0\approx 0$.
 \end{lemma}
 \begin{proof}
    {  For the sake of simplicity, we focus on the HD mode but the analysis  also holds true for the FD mode.
    We prove this lemma from optimization's viewpoint by comparing  DF
    problem (\ref{eqn:prob:rate:max:DF:v0}) with AF problem (\ref{eqn:prob:rate:max:AF:v1}).

    Given an AF optimal beamforming solution    $(\qw_a^*,\qw_c^*)$ to    (\ref{eqn:prob:rate:max:AF:v1}),  we construct a new solution
    $(\qw_0^*,\qw_c^*)$ where   $\qw_0^* = \sqrt{P_0\|\qg\|^4 + \|\qg\|^2} \qw_a^*$ and then check whether it is feasible for
    (\ref{eqn:prob:rate:max:DF:v0}). It is   seen that $(\qw_0^*,\qw_c^*)$ achieves the same objective value for (\ref{eqn:prob:rate:max:DF:v0}) as
    $(\qw_a^*,\qw_c^*)$ for  (\ref{eqn:prob:rate:max:AF:v1}) and satisfies the last power constraint in
    (\ref{eqn:prob:rate:max:DF:v0}).  The second constraint in
    (\ref{eqn:prob:rate:max:DF:v0}) can be verified below
     \bea
       && \frac{|\qh_{c0}^\dag  \qw_0|^2}{|\qh_{c0}^\dag  \qw_c|^2+ P_t\|\qh_{c0}\|^2 +1}\notag\\
        &=&\frac{|\qh_{c0}^\dag  \qw_a|^2 (P_0\|\qg\|^4 + \|\qg\|^2)}{|\qh_{c0}^\dag  \qw_c|^2+ P_t\|\qh_{c0}\|^2 +1}\notag\\&\ge & \frac{\gamma_{0_{AF}}^{'}(P_0\|\qg\|^4 +
        \|\qg\|^2)}{(\|\qg\|^4-\gamma_{0_{AF}}^{'}\|\qg\|^2)}\label{eqn:comp:AF:DF:1}\\
        &=& \frac{ \frac{\gamma_{0_{DF}}^{'}}{P_0}(P_0\|\qg\|^4 + \|\qg\|^2)}{(\|\qg\|^4-\gamma_{0_{AF}}^{'}\|\qg\|^2)}\notag\\
      & =& \frac{  \gamma_{0_{DF}}^{'} (\|\qg\|^4 + \frac{\|\qg\|^2}{P_0})}{(\|\qg\|^4-\gamma_{0_{AF}}^{'}\|\qg\|^2)}>\gamma_{0_{DF}}^{'},\label{eqn:comp:AF:DF:2}
     \eea
    where (\ref{eqn:comp:AF:DF:1}) comes from the fact that $(\qw_a^*,\qw_c^*)$ is an optimal solution to    (\ref{eqn:prob:rate:max:AF:v1}) and we
    have used the approximation $h_0\approx 0$ and $\frac{\gamma_{0_{DF}}^{'}}{P_0} = \gamma_{0_{AF}^{'}}$ in (\ref{eqn:comp:AF:DF:2}).

    We have proved that for any optimal solution to AF problem (\ref{eqn:prob:rate:max:AF:v1}), we can find a feasible solution to DF problem
    (\ref{eqn:prob:rate:max:DF:v0}) with the same objective value, therefore the optimal solution to (\ref{eqn:prob:rate:max:DF:v0}) should have a
    greater objective value or the  DF relaying protocol achieves
    higher CU rate than the   AF relaying protocol.}
     \end{proof}

The assumption that the primary direct link is weak is one of the rationales in order to enable a cooperation between primary and secondary sources:
the primary link is weak with respect to the channel from the primary transmitter to the secondary transmitter
\cite{Simeone-07-stability}\cite{Krikidis-10} or it is not available due to path-loss effects and physical obstacles \cite{Bletsas-06}.
 The assumption $h_0\approx 0$ will help  derive a simple mode switching criterion and  Lemma 1
cannot be generalized for all cases without this assumption.

  It is also shown in \cite{AF-DF-appro} that asymptotically  the performance of AF is quite close to that of DF. Therefore, we choose the DF
  protocol to  compare the performance of the HD and FD modes for simplicity. We adopt the comparison result as  a unified criterion to select the mode
  using both AF and DF protocols. Next we will derive simpler closed-form DF solutions for the HD and FD modes   based on ZF criterion. For this  derivation,
  we assume that the transmit noise power scales with the  signal power, i.e.,  $P_t = \epsilon^2p_R $ and $P_t = \epsilon^2p_{R_{FD}}$ for HD and FD modes, respectively.
 \subsection{DF in the HD mode}
  For convenience, we define $\rho=1-\frac{|\qh_{c}^\dag\qh_{c0}|^2}{\|\qh_{c}\|^2\|\qh_{c0}\|^2}$ and write the beamforming vectors in the form:
 $\qw_c = \sqrt{q_c} \bar\qw_c, \qw_0 =\sqrt{q_0} \bar\qw_0$ with $\|\bar \qw_c\|=\|\bar \qw_0\|=1$ where $q_c$ and $q_0$ denote the transmit power
 for primary and cognitive signals, respectively.

 According to the ZF criterion, we have $\qw_c^\dag\qh_{c0}=0$ and in addition, $\qw_c$ needs to  maximize $|\qw_c^\dag\qh_{c}|$, therefore it admits the following expression
 \be
     \qw_c = \sqrt{q_c} \frac{ \left(\qI - \frac{\qh_{c0}\qh_{c0}^\dag}{\|\qh_{c0}\|^2}\right) \qh_c}{\|\left(\qI - \frac{\qh_{c0}\qh_{c0}^\dag}{\|\qh_{c0}\|^2}\right)
    \qh_c\|},
 \ee
 and the resulting CU channel gain is
  \be
    |\qw_c^\dag\qh_{c}|^2 = q_c\|\qh_{c}\|^2(1-\rho^2).
 \ee
 Similarly,

\vspace{-10pt}
{\small
\be
    \qw_0 = \sqrt{q_0}\frac{ \left(\qI - \frac{\qh_{c}\qh_{c}^\dag}{\|\qh_{c}\|^2}\right) \qh_0}{\| \left(\qI - \frac{\qh_{c}\qh_{c}^\dag}{\|\qh_{c}\|^2}\right)
    \qh_0\|}, \mbox{~and~}      |\qw_0^\dag\qh_{c0}|^2 =q_0\|\qh_{c0}\|^2(1-\rho^2).
 \ee
}%
Based on the above expressions,  the DF problem formulation (\ref{eqn:prob:rate:max:DF:v0}) is simplified  to the following power allocation:
  \bea\label{eqn:prob:DF:HD:power:0}
    \min_{q_0,q_c} &&  \frac{q_c\|\qh_c\|^2(1-\rho^2)}{ \epsilon^2P_C\|\qh_c\|^2+1} \\
    \mbox{s.t.} &&\frac{q_0\|\qh_{c0}\|^2(1-\rho^2)}{  \epsilon^2 P_C \|\qh_0^\dag\|^2 +1} \ge \gamma_{0_{DF}}^{'},\notag\\
&&  P_0\|\qg\|^2 \ge \gamma_{0_{DF}}^{'},  q_c + q_0\le  P_C,\notag
 \eea
 which gives the CU rate below

\vspace{-10pt}
{\small
 \be\label{eqn:CU:rate:DF:HD}
   \scriptstyle  R_C = \left\{\begin{array}{ll}
  \frac{1}{2}\log_2\left(1+   \frac{\max\left(0, P_C(1-\rho^2) -  \gamma_{0_{DF}}^{'}( \epsilon^2   P_C  +\frac{1}{\|\qh_{c0}\|^2})\right)}{ \epsilon^2
 P_C+\frac{1}{\|\qh_c\|^2}} \right), &\\
  \mbox{for }P_0\|\qg\|^2 \ge\gamma_{0_{DF}}^{'}&;\\
 0, ~~~~  \mbox{otherwise}. &
 \end{array}\right.
 \ee
}%
  Ignoring the PBS-PU and PBS-CU links, we have the following  approximation when $P_0\|\qg\|^2 \ge\gamma_{0_{DF}}^{'}$:

\vspace{-10pt}
{\small
   \bea
\textstyle R_C \approx \frac{1}{2}\log_2 \left(1 +      \frac{\max\left(0,  1-\rho^2 - (2^{2r_0}-1)(\epsilon^2     +
  \frac{    1}{  P_{C}\|\qh_{c0}\|^2 })\right)}{  \epsilon^2       +  \frac{  1}{   P_{C}
  \|\qh_{c}\|^2}}\right).
 \eea
}

\subsection{DF in the FD mode}
 Similar to the HD mode, we define $\rho_{FD}=1-\frac{|\qh_{c_{FD}}^\dag\qh_{c0_{FD}}|^2}{\|\qh_{c_{FD}}\|^2\|\qh_{c0_{FD}}\|^2}$, $\qw_c = \sqrt{q_c}
 \bar\qw_c, \qw_0 =\sqrt{q_0} \bar\qw_0$ where $\|\bar \qw_c\|=\|\bar \qw_0\|=1$ and $q_c$ and $q_0$ denote the respective transmit power for primary
 and cognitive signals, respectively. The DF problem in (\ref{eqn:prob:SINR:FD:DF:v0}) for the FD mode becomes
  \bea
    \min_{q_0,q_c} &&  \frac{ q_c\|\qh_{c_{FD}}\|^2(1-\rho_{FD}^2)}{     \epsilon^2\bar P_{C_{FD}}\|\qh_{c_{FD}}\|^2       +  P_{0_{FD}} |{h}_{0c}|^2   + 1} \\
    \mbox{s.t.} 
  &&  \frac{q_0\|\qh_{c0_{FD}}\|^2(1-\rho_{FD}^2)}{   \epsilon^2 \bar P_{C_{FD}} \|\qh_{c0_{FD}}\|^2      +  P_{0_{FD}} |{h}_{0}|^2  + 1 } \ge  \gamma_{0_{FD}}\notag\\
    && q_c + q_0  \le \bar P_{C_{FD}},\notag
  \eea 
 and   gives the CU rate (43).
\begin{figure*}
 \bea\label{eqn:rate:CU:DF:FD}
  \scriptstyle  R_{C_{FD}} =
   \log_2 \left(1 +      \frac{\max\left(0, \bar P_{C_{FD}}(1-\rho_{FD}^2)-\gamma_{0_{FD}}(\epsilon^2  \bar P_{C_{FD}}     +  \frac{P_{0_{FD}} |{h}_{0}|^2  +
    1}{\|\qh_{c0_{FD}}\|^2 })\right)}{  \epsilon^2 \bar P_{C_{FD}}        +  \frac{P_{0_{FD}} |{h}_{0c}|^2   +
    1}{\|\qh_{c_{FD}}\|^2}}\right).
 \eea
\end{figure*}
  Ignoring the PBS-PU and PBS-CU links, we have the approximation (44).
\begin{figure*}
   \bea
  \scriptstyle  R_{C_{FD}} \approx
 \log_2 \left(1 +      \frac{\max\left(0,  \bar P_{C_{FD}}(1-\rho_{FD}^2) -(2^{r_0}-1)(\bar P_{C_{FD}}\epsilon^2     +
  \frac{1}{\|\qh_{c0_{FD}}\|^2 })\right)}{  \bar P_{C_{FD}}\epsilon^2       +  \frac{  1}{  \|\qh_{c_{FD}}\|^2}}\right).
 \eea
\end{figure*}

 The mode selection for both AF and DF relaying protocols corresponds to a simple comparison between the achievable rates in (\ref{eqn:CU:rate:DF:HD}) and
 (\ref{eqn:rate:CU:DF:FD}).

In the next two subsections, we assume that the PBS-CBS links $\|\qg\|^2$ and $\|\qg_{FD}\|^2$ are sufficiently strong to support the required PU
rate and we focus on the CBS-PU and CBS-CU links to gain some insights on the impact of some system parameters.

\subsection{Same RF chains for HD and FD}
 First we assume that both HD and FD modes have the same number of RF chains, and the same sets of
 transmit and receive antennas, so we remove the subscript `HD'. In this case, all
 corresponding channel matrices are the same for both modes and the
 achievable CU rates are

\vspace{-10pt}
{\small
 \be
   \textstyle  
    \begin{aligned}
    &R_{C} \approx\\
     &\frac{1}{2}\log_2 \left(1 +      \frac{\max\left(0,  1-\rho^2 - (2^{2r_0}-1)(\epsilon^2     +
  \frac{    1}{  P_{C}\|\qh_{c0}\|^2 })\right)}{  \epsilon^2       +  \frac{  1}{   P_{C}
  \|\qh_{c}\|^2}}\right),  \\
&    R_{C_{FD}} \approx\\
 &    \log_2  \left(1 +      \frac{\max\left(0,  1-\rho^2 -(2^{r_0}-1)(\epsilon^2     +
  \frac{    2}{  P_{C}\|\qh_{c0}\|^2 })\right)}{  \epsilon^2       +  \frac{  2}{  P_{C}
  \|\qh_{c}\|^2}}\right).
  \end{aligned}
 \ee
}

 Setting both rates  to be equal to zero (if possible), we get the corresponding zero points for $\epsilon^2$, which represent the
 maximum tolerable transmit noise factors:
\be
    \begin{aligned}
    \epsilon_{FD}^2 = \frac{1-\rho^2}{ (2^{r_0}-1)} - \frac{2}{P_C\|\qh_{c0}\|^2},\\~
    \epsilon_{HD}^2 = \frac{1-\rho^2}{ (2^{2r_0}-1)} -
    \frac{1}{P_C\|\qh_{c0}\|^2}.
  \end{aligned}
 \ee
 We then derive the difference and the relative difference:
  \be
   \begin{aligned}
      \epsilon_{0,FD}^2 -  \epsilon_{0,HD}^2 &=& (1-\rho^2)  \frac{2^{r_0}}{ (2^{2r_0}-1)}  -
      \frac{1}{P_C\|\qh_{c0}\|^2},\\
     \frac{\epsilon_{0,FD}^2 -  \epsilon_{0,HD}^2}{\epsilon_{0,HD}^2} &=& \frac{(1-\rho^2) \frac{2^{r_0}}{ (2^{2r_0}-1)}  -
     \frac{1}{P_C\|\qh_{c0}\|^2}}{     \frac{1-\rho^2}{ (2^{2r_0}-1)} -
    \frac{1}{P_C\|\qh_{c0}\|^2}}\\ &=& 2^{r_0} + \frac{ (2^{r_0}-1)\frac{1}{P_C\|\qh_{c0}\|^2} }{\frac{1-\rho^2}{ (2^{2r_0}-1)} -
    \frac{1}{P_C\|\qh_{c0}\|^2}}.
  \end{aligned}
 \ee
 It is observed that as $P_C$ increases or $\rho^2$ decreases, FD  tends to be better than HD, but the relative improvement becomes less.
 As $r_0$ is decreasing, HD tends to perform better than FD. 

 \subsection{High CBS power $P_C$}
  In this case, HD and FD can have different RF chains and sets of   transmit and receive antennas. We focus on the scenario where the
  CBS's power is large. By using this assumption, the approximations of the  achievable rates are (48) and (49)
\begin{figure*}
  {\small
 \bea
 \scriptstyle   &&R_C \approx  
  \frac{1}{2}\log_2 \left(1 +      \frac{\max\left(0,  1-\rho^2 - (2^{2r_0}-1)(\epsilon^2     +
  \frac{    1}{  P_{C}\|\qh_{c0}\|^2 })\right)}{  \epsilon^2       +  \frac{  1}{   P_{C}
  \|\qh_{c}\|^2}}\right)
\approx  \frac{1}{2}\log_2\left(1+ \max\left(0,\frac{  1-\rho^2 }{ \epsilon^2} +2 -
 2^{2r_0}\right)\right),
\eea
 \bea
&&  R_{C_{FD}}  \approx
  \log_2 \left(1 +      \frac{\max\left(0,  1-\rho_{FD}^2 -(2^{r_0}-1)(\epsilon^2     +
  \frac{2}{P_C\|\qh_{c0_{FD}}\|^2 })\right)}{  \epsilon^2       +  \frac{  2}{ P_C
  \|\qh_{c_{FD}}\|^2}}\right)
\approx   \log_2\left(1+ \max\left(0, 2+ \frac{  1-\rho_{FD}^2 }{ \epsilon^2} -
 2^{r_0}\right)\right). 
 \eea}
\end{figure*}
 When $r_0>1$, the zero points for $\epsilon^2$ are
 \be
     \begin{aligned}\epsilon_{FD}^2 &= \frac{  1-\rho_{FD}^2 }{ 2^{r_0} - 2 },~~
    \epsilon_{HD}^2 &= \frac{  1-\rho^2 }{ 2^{2r_0} - 2 }.
 \end{aligned}
 \ee
 We can see that when
 \be
    \rho_{FD} \le   \frac{2^{2r_0} - 2^{r_0}  }{ 2^{2r_0} - 2 } + \frac{2^{r_0} - 2}{ 2^{2r_0} - 2 }\rho_{HD}^2,
 \ee
  the FD mode performs better than the HD mode. This happens when $N_t$ is close to $N$ or both are
  large.
\vspace{-5mm}
\section{Numerical Results}
Computer simulations are conducted to evaluate the performance of the proposed FD and hybrid schemes. We assume that the CBS has $N=6$ antennas,
which is a sufficient configuration for demonstrating the performance of the investigated schemes. The channel between any antenna pair from
different terminals is modeled as $h=d^{-\frac{c}{2}}e^{j\theta}$, where $d$ is the distance, $c$ is the path loss exponent (chosen as $3.5$), and
$\theta$ is uniformly distributed over $[0,2\pi)$. The distances from the CBS to the PBS, the PU and the CU are all normalized to one unit while the
distances from the PBS to the PU, the PBS to the CU are set to $2$ units, so their channels are much weaker than other links. The elements of the
loop interference channel $\qH$ are independent and identically distributed following $\mathcal{CN}(0,1)$. Unless otherwise specified, the PU's
target rate is $r_0=2$ bits per channel use (bpcu); for the HD mode, the CBS can use all $6$ antennas for receiving and transmitting signals; for the
FD mode, the CBS uses $N_t=4$ transmit antennas and $N_r=2$ receive antennas; the CBS transmit noise power is assumed to scale with the CBS power and
$\epsilon^2=10^{-4}$. We define the transmit SNR in the HD mode, transmit power normalized by noise power, as the power metric and the primary
transmit power is set to $10$ dB. Half transmit power in the FD mode is used to maintain the same energy consumption as the HD mode. An outage event
occurs when $r_0$ is not supported in the primary system for a channel instance. Except for Figs \ref{fig:rate:region:AF} and
\ref{fig:rate:region:DF}, $10^3$ and $10^4$ channel realizations are used to produce the results of the average rate and outage probability,
respectively. Whenever possible and necessary, the proposed FD and hybrid schemes will be compared with  the
 following solutions:
 \begin{itemize}
    \item DPC with a non-causal primary message at the CBS \cite{Devroye-bound}, \cite{Jafar-bound},\cite{Viswanath-bound}.
    In this case, the CBS uses the principles of DPC and pre-cancels the non-causal primary information in order to ensure an interference-free secondary
transmission. This scheme  requires a non-causal knowledge of the primary message at the CBS and therefore it has a limited practical interest.
However, it provides a useful theoretical upper-bound for any practical cognitive cooperative scheme and can be used for comparison purposes;
    \item Orthogonal transmission, i.e., the CBS transmits in such a
    way  to not interfere the PU without assisting the PU's transmission.
    The CBS rate can be found by solving  the following optimization problem:
       \bea
  \max_{ \qw_c}&& \log_2\left(1+\frac{|\qh_{c}^\dag  \qw_c|^2}{P_t \|\qh_c\|^2+1} \right) \notag\\
    \mbox{s.t.}&&
    \qh_{c0}^\dag  \qw_c=0, \|\qw_c\|^2 \le P_C;
 \eea
    \item HD mode, by default we assume that all $N$ antennas are used for both transmission and reception;
    \item HD mode using the same RF chains as the FD mode;
    \item Best HD/FD mode selection.
 \end{itemize}
 In Figs \ref{fig:rate:region:AF} and \ref{fig:rate:region:DF}, we plot the rate regions for a specific channel realization for AF and DF,
 respectively, when the CBS power is $10$ dB. It can be verified  that DPC provides a performance outer-bound and the performance difference between the proposed scheme and the DPC is mainly due
to the unrealistic assumption of non-causal primary information for the DPC.  The orthogonal transmission achieves a larger
 rate region for the CU because it does not assist the PU's transmission. Even when the HD mode uses the same RF chains as the FD mode, i.e., $4$ transmit
 antennas and $2$ receive antennas, the maximum PU rate is over three times higher than  that of  the orthogonal scheme. Further improvement is
 observed when the CBS uses all the $6$ antennas for both transmission and reception in the HD mode. When the CBS works in the FD mode, approximately
 $50\%$ higher rates for both the PU and the CU are achieved, compared with the HD mode using the same RF chains.
\begin{figure}[t]
\centering
\includegraphics[width=2.5in]{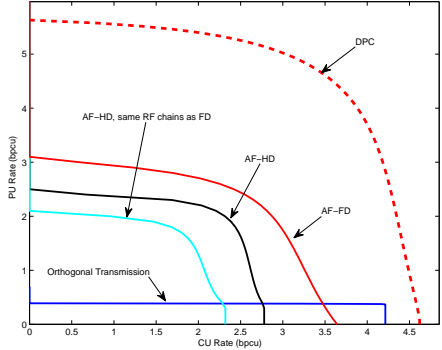}
\caption{Rate region for AF protocol.}\label{fig:rate:region:AF}
\end{figure}
\begin{figure}[t]
\centering
\includegraphics[width=2.5in]{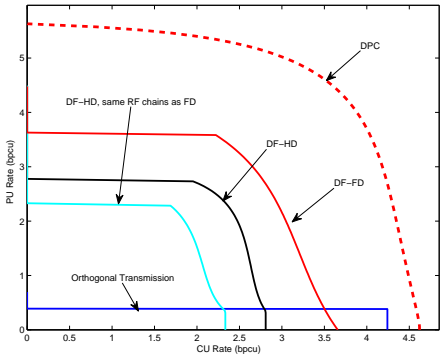}
\caption{Rate region for DF protocol.}\label{fig:rate:region:DF}
\vspace{-5mm}
\end{figure}

In Fig. \ref{fig:cu:rate:vs:cbs:power}, we plot the CU rate against the CBS power. It can be seen that the proposed FD and hybrid
schemes achieve almost three times the  CU rates provided by the HD mode with  the same RF chains. At low SNRs, the HD mode may perform better than
the FD mode while the performance gain of the FD over the HD mode is enlarged as the CBS power increases. It is also observed that the proposed hybrid
schemes perform nearly as well as the best mode selection. The proposed FD and hybrid schemes achieve the same slope for the CU rate as DPC.
\begin{figure}[t]
\centering
\includegraphics[width=2.5in]{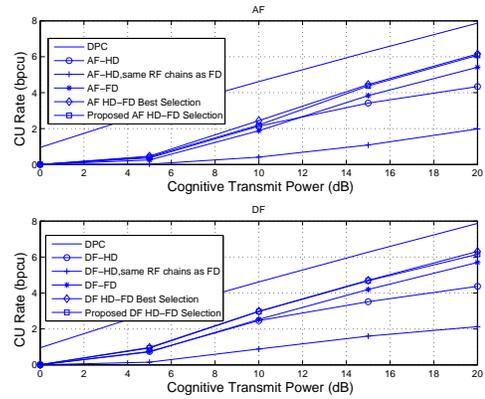}
\caption{CU rate vs. CBS Power.}\label{fig:cu:rate:vs:cbs:power}
\vspace{-5mm}
\end{figure}

In Fig. \ref{fig:cu:rate:vs:scale:nonscale}, we   study the impact of fixed and scalable transmit noise on the CU rate for the FD mode. It can be
observed that with a fixed transmit noise power, all the achievable rates are increasing with the CBS power. The achievable rate of the AF protocol
approaches that of the DF protocol at high SNRs. For the DPC scheme and the HD mode, the achievable rates saturate from $25$ dB of the CBS power,
which indicates that the CBS  should reserve some power in order to suppress the self-interference.
\begin{figure}[t]
\centering
\includegraphics[width=2.5in]{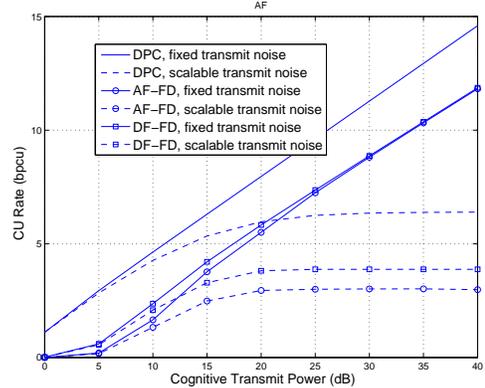}
\caption{CU rate for both scalable and nonscalable transmit noise power.}\label{fig:cu:rate:vs:scale:nonscale}
\vspace{-5mm}
\end{figure}

In Fig. \ref{fig:cu:rate:vs:noise:power}, we plot the CU rates versus the transmit noise scaling factor $\epsilon^2$ , when the CBS power is
$20$ dB.  We note that in general, FD outperforms HD even when $\epsilon^2$ is large, especially for the DF relaying protocol. This observation is
because the transmit noise limits the performance of HD while  FD can efficiently suppress the self-interference by employing optimal relay processing.
As for the AF relaying protocol, the HD mode may outperform the FD mode due to the fact that the CBS amplifies the self-interference. The performance of the proposed mode
selection is very close to the best selection for the AF relaying protocol; for the DF relaying protocol, they are almost identical and validate the effectiveness of the proposed selection.

\begin{figure}[t]
\centering
\includegraphics[width=2.5in]{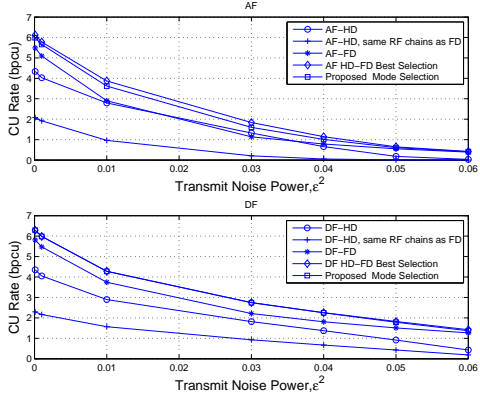}
\caption{CU rate vs. transmit noise power.}\label{fig:cu:rate:vs:noise:power}
\vspace{-5mm}
\end{figure}

In Fig. \ref{fig:pu:outage:vs:cbs:power},  we examine the effects of the CBS power on the PU's outage performance. We assume the PU requires a
rate of $3$ bpcu. First due to the weak primary link, the outage is almost $100\%$ if there is no assistance from the CBS. As expected, the HD mode is
not efficient  at low SNRs, and can even become worse than the direct transmission due to the two phases used; for AF, the outage performance is
improved only when the CBS power is higher than $15$ dB, while with the same energy consumption, FD and the hybrid schemes can reduce the outage
probability to $60\%$. With $25$ dB for the CBS power, AF-HD has an outage probability about $60\%$,  while the hybrid schemes reduces the outage probability
to below $40\%$. The FD mode with DF protocol achieves a lower saturated outage probability of $32\%$ when the CBS power is above $20$ dB.
\begin{figure}[t]
\centering
\includegraphics[width=2.5in]{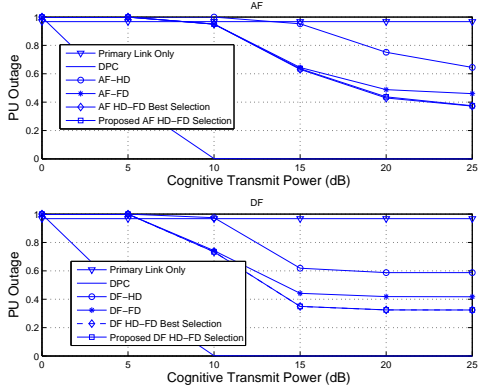}
\caption{PU rate outage vs. CBS power.}\label{fig:pu:outage:vs:cbs:power}
\vspace{-5mm}
\end{figure}

Finally in Fig. \ref{fig:antenna:config},  we investigate the impact of different transmit and receive antenna configurations on the FD mode, by
assuming the CBS has $6$ transmit and receive RF chains. We simulate four different cases $(N_t, N_r)\in \{(5,1),(4, 2),(3,3),(2,4)\}$. The same
trends are observed for both AF and DF protocols. It can be seen that $(N_t,N_r) =(3,3)$ provides the best performance for medium to high SNRs; this
is because the primary rate is upper bounded by the supportable rates of the PBS-CBS and CBS-PU links and therefore equal number of transmit and
receive antennas is a preferred configuration. At low to medium SNRs,  $(N_t,N_r) =(5,1)$ provides better performance than that of $(N_t,N_r)=(2,4)$,
because the rate is limited by the CBS-PU and CBS-CU links. Therefore the CU rate takes the expression (53). 
\begin{figure*}
{\small
    \bea
    &&R_{C_{FD}} \approx
\log_2 \left(1 +      \frac{\max\left(0,  1-\rho_{FD}^2 -(2^{r_0}-1)(\epsilon^2     +
  \frac{2}{P_C\|\qh_{c0_{FD}}\|^2 })\right)}{  \epsilon^2       +  \frac{  2}{ P_C
  \|\qh_{c_{FD}}\|^2}}\right).
 \eea}
\end{figure*}
 It is observed that as $N_t$ increases, both $\|\qh_{c0_{FD}}\|^2$ and $\|\qh_{c_{FD}}\|^2$ increase while $\rho_{FD}^2$ decreases, consequently the
 CU rate increases as well.  At high SNRs,  $(N_t,N_r) =(2,4)$ outperforms the $(N_t,N_r) =(5,1)$ configuration. This is because   the
 PBS-CBS link   limits the achievable  PU and CU rates and therefore more received antennas at the CBS can improve the CU rate. For the same reason,
 $(N_t,N_r) =(4,2)$ saturates when the CBS power is about $20$ dB while the performance of the case   $(N_t,N_r)
 =(2,4)$ improves with the CBS power until $25$ dB.
\begin{figure}[t]
\centering
\includegraphics[width=2.5in]{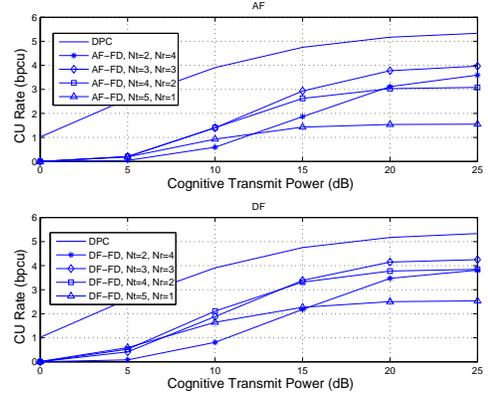}
\caption{CU rate for different antenna configurations}\label{fig:antenna:config}
\vspace{-5mm}
\end{figure}

\section{Conclusion}
 We have studied the HD and the FD operation modes for the CBS in  a   cooperative cognitive network. We have considered transmit imperfections  in both duplex modes and  modeled the resulting CBS's  residual self-interference  for the FD mode. Closed-form solutions or efficient 1-D search algorithms to achieve the optimal AF and DF beamforming vectors  have been provided in order to characterize the achievable primary-cognitive rate regions. In addition, we proposed a hybrid scheme to switch between the HD and FD modes based on the simplified ZF beamforming design.
Results have shown that the proposed FD and hybrid schemes can greatly enlarge the rate region compared to the HD mode, therefore they are introduced as
efficient solutions for the active cooperation between primary and cognitive systems. The proposed cooperation  substantially increases the opportunities for a CU to
access the primary spectrum and improves the overall system spectral efficiency.

It is worth noting that for scenarios with a strong primary direct link,  the PU receives two copies of the transmitted signal via both the direct and the relaying links by generating  an artificial multipath effect. While this paper and most studies in the literature discard the direct link and only decode the relaying information, as a future direction,  we can study how to efficiently combat this effect using equalization techniques \cite{outage-AF-FD}.

\section*{\sc Appendices}

\section*{A. Proof of Theorem 1}
\begin{proof}
 Without loss of generality, $\qA$ can be expressed in the form of
 \bea
     \qA &=& [\bar\qH  ~~\bar\qH^\bot] \left[\begin{array}{cc}
                               \qb & \qC \\
                               \qd & \qE
                             \end{array}\right] [\qg^\dag~~
                             ({\qg^\bot})^\dag] \\
                             &=& \bar\qH \qb\qg^\dag +
                             \bar\qH^\bot \qd \qg^\dag + \bar\qH
                             \qC({\qg^\bot})^\dag +
                             \bar\qH^\bot\qE({\qg^\bot})^\dag\notag
 \eea
where $\bar\qH\triangleq [\qh_{c0} ~\qh_c]$ and  $\qb \in \mathbb{C}^{2\times 1}, \qC \in \mathbb{C}^{2\times(N-1)}, \qd\in\mathbb{C}^{(N-2)\times
1}, \qE\in\mathbb{C}^{(N-2)\times(N-1)}$ are parameter vectors and matrices.

 A closer observation of the defined problem  reveals that the optimization will maximize    $|\qh_{c0}^\dag\qA\qg|^2$ while
minimize $|\qh_c^\dag\qA\qg|^2$, $\|\qh_c^\dag\qA\|^2, \|\qh_{c0}\qA\|^2, \|\qA\qg\|^2$ and $\|\qA\|^2$.  It is clearly seen that $\qC,\qE,  \qd$ do
not affect the term to be maximized and setting them to be zero will help reduce the terms to  be minimized, therefore $\qC=\qzero, \qd =
\qzero,\qE=\qzero$ and the optimal $\qA$ has the structure of $\qA = \bar\qH\qb\qg^\dag$ and can be written in a more general form of $\qA =
\qw_a\qg^\dag$ where $\qw_a=\bar\qH\qb$ is a new parameter vector.
\end{proof}

 \section*{B. Closed-form Solution to A General Rate Maximization Problem}
  Our aim here is to find the close-form solution to the rate maximization problem below:
    \bea\label{eqn:prob:rate:max:basic}
    \max_{\qw_1,\qw_2} && \frac{|\qh_2^\dag\qw_2|^2}{1+
    |\qh_2^\dag\qw_1|^2}\\
    \mbox{s.t.} && \frac{|\qh_1^\dag\qw_1|^2}{1+
    c|\qh_1^\dag\qw_2|^2}\ge \gamma_1, |\qw_1\|^2+ c\|\qw_2\|^2\le P_C,\notag
    \eea
    where $c$ is a constant, $\qh_1,\qh_2$ are $N\times 1$ vectors and $\gamma_1, P_C$ are positive scalars.{ This problem has the following physical
    meaning. Consider a   MISO broadcast system with an $N$-antenna BS and two single-antenna users. The channels from the BS to user 1 and user 2
    are $\qh_1$ and $\qh_2$, respectively. The noise powers at users are assumed to be one, otherwise, the channel can be normalized with the noise
    power.  Suppose the BS has a total power constraint $P_C$ and user 1 has a SINR constraint $\gamma_1$, then this problem has the interpretation
    of maximization of   user 2's SINR.}
   Suppose its optimal objective value is $\gamma_2^*$. To find the optimal solution to (\ref{eqn:prob:rate:max:basic}), we first consider the following
  weighted sum power minimization problem:
    \bea\label{eqn:prob:pow:min:basic}
    \min_{\qw_1,\qw_2} && \|\qw_1\|^2+ c\|\qw_2\|^2\\
    \mbox{s.t.} && \frac{|\qh_1^\dag\qw_1|^2}{1+
    c|\qh_1^\dag\qw_2|^2}\ge \gamma_1, \frac{|\qh_2^\dag\qw_2|^2}{1+
    |\qh_2^\dag\qw_1|^2}\ge \gamma_2.\notag
    \eea
  It can be validated that if we set $\gamma_2=\gamma_2^*$ in (\ref{eqn:prob:pow:min:basic}), then its optimal objective value is $P_C$ and vice
  versa. So we can focus on (\ref{eqn:prob:pow:min:basic}) in order to characterize the
  solution to (\ref{eqn:prob:rate:max:basic}).   The dual problem of (\ref{eqn:prob:pow:min:basic}) can be derived as
   \bea
    \max_{\lambda_1,\lambda_2\ge 0} && \lambda_1+\lambda_2\\
    \mbox{s.t.} && \qI +\lambda_2\qh_2\qh_2^\dag\succeq
    \frac{\lambda_1}{\bar\gamma_1}\qh_1\qh_1^\dag, \bar\gamma_1\triangleq   c\gamma_1\notag\\
    &&\qI +\lambda_1\qh_1\qh_1^\dag\succeq
    \frac{\lambda_2}{\gamma_2}\qh_1\qh_2^\dag\notag,
   \eea
   where $\lambda_1$  and $\lambda_2$ are dual variables.   The two linear matrix inequality constraints uniquely determine  $\lambda_1$ and $\lambda_2$:

\vspace{-10pt}
{\small
   \be
    \begin{aligned}
 \lambda_1&=& \frac{\bar\gamma_1}{ \qh_{1}^\dag\left(\qI + \lambda_1 \bar \qh_2\qh_2^\dag \right)^{-1}\qh_{1}},
 \lambda_2 &=& \frac{\gamma_2}{ \qh_{2}^\dag \left(\qI +
\lambda_1\qh_{1}\qh_{1}^\dag  \right)^{-1}
         \qh_2}.
        \end{aligned}\label{eqn:lambda01} 
 \ee
 }%
 Using matrix inversion lemma and define $\rho^2 \triangleq \frac{|\qh_1^\dag\qh_2|^2}{\|\qh_1\|^2\|\qh_2\|^2}$ we have
\be
    \begin{aligned}
  \lambda_1&=&\frac{\bar\gamma_1(1+\lambda_2\|\qh_2\|^2)}{\|\qh_{1}\|^2
+ \lambda_2(\|\qh_{2}\|^2\|\qh_1\|^2 -
    \|\qh_{2}^\dag\qh_1\|^2)}\\
    &=&\frac{\bar\gamma_1(1+\lambda_2\|\qh_2\|^2)}{\|\qh_1\|^2(1+\lambda_2\|\qh_2\|^2(1-\rho^2))}\label{eqn:lambda:c},\\
      \lambda_2&=&\frac{\gamma_2(1+\lambda_1\|\qh_1\|^2)}{\|\qh_{2}\|^2 + \lambda_1(\|\qh_{2}\|^2\|\qh_1\|^2 -
    \|\qh_{2}^\dag\qh_1\|^2)}\\ &=&\frac{\gamma_2(1+\lambda_1\|\qh_1\|^2)}{\|\qh_2\|^2(1+\lambda_1\|\qh_1\|^2(1-\rho^2))}.
\end{aligned}
 \ee Remember we also have a power equation below:
 \be\label{eqn:lambda01:sum}
    \lambda_1+\lambda_2 = P_C.
 \ee
 It is observed that $\lambda_1,\lambda_2,\gamma_2$ should satisfy and uniquely determined by  the above three equations
(\ref{eqn:lambda:c}-\ref{eqn:lambda01:sum}), so the analytical solutions can be found. Define $A=\|\qh_1\|^2 \|\qh_2\|^2(1-\rho^2),  B= -
(\|\qh_2\|^2\bar\gamma_1 + P_C  \|\qh_1\|^2 \|\qh_2\|^2(1-\rho^2) +  \|\qh_1\|^2)$, and $C= (P_C \|\qh_2\|^2+1)\bar\gamma_1$.
 Then from (\ref{eqn:lambda:c}), we have
 \be
    f(\lambda_1) \triangleq A\lambda_1^2 + B\lambda_1+C=0.
 \ee
 Since $A>0, C>0, B<0$,  $f(\lambda)=0$ has positive two roots. Because $f(P_C)<0$, we know that the optimal $\lambda_1$ corresponds to the minimum
 root. Once $\lambda_1$ is found, $\lambda_2$ and $\gamma_2$ can be easily derived from  (\ref{eqn:lambda:c}-\ref{eqn:lambda01:sum}).

 \end{document}